\documentclass[conference]{IEEEtran}

\usepackage[colorlinks,
linkcolor=red,
anchorcolor=green,
citecolor=blue
]{hyperref} 
\usepackage{amsfonts,amssymb,amsmath,array}
\usepackage{latexsym}
\usepackage{CJK}
\usepackage{cite}
\usepackage{bm}
\usepackage{color, soul}
\usepackage{colortbl}

\usepackage{graphicx}
\usepackage{epstopdf}
\usepackage{epsfig}
\usepackage{subfigure} 
\usepackage{framed}
\usepackage{verbatim}  
\usepackage{color,soul}
\definecolor{aliceblue}{rgb}{0.94, 0.97, 1.0}
\definecolor{blizzardblue}{rgb}{0.67, 0.9, 0.93}
\definecolor{antiquebrass}{rgb}{0.8, 0.58, 0.46}
\definecolor{beaublue}{rgb}{0.74, 0.83, 0.9}
\sethlcolor{blizzardblue}
\usepackage{mathrsfs}
\usepackage{algorithm} 
\usepackage{algorithmic} 
\usepackage{booktabs}
\usepackage{textcomp}
\usepackage{multirow}
\usepackage[scaled=0.92]{helvet}

\usepackage{algorithm}
\usepackage{algorithmic}

\usepackage{tcolorbox}

\usepackage{amsthm} 
\newtheorem{lemma}{Lemma}

\begin{document}

\title{Towards Big data processing in IoT: network management for online edge data processing}

\author{Shuo~Wan, Jiaxun~Lu, Pingyi~Fan,~\IEEEmembership{Senior Member,~IEEE} and Khaled~B.~Letaief{*},~\IEEEmembership{Fellow,~IEEE}\\

\small
Tsinghua National Laboratory for Information Science and Technology(TNList),\\

Department of Electronic Engineering, Tsinghua University, Beijing, P.R. China\\
E-mail: wan-s17@mails.tsinghua.edu.cn, lujx14@mails.tsinghua.edu.cn, ~fpy@tsinghua.edu.cn\\
{*}Department of Electronic Engineering, Hong Kong University of Science and Technology, Hong Kong\\
Email: eekhaled@ece.ust.hk}

\maketitle

\begin{abstract}
Heavy data load and wide cover range have always been crucial problems for internet of things (IoT). However, in mobile-edge computing (MEC) network, the huge data can be partly processed at the edge. In this paper, a MEC-based big data analysis network is discussed. The raw data generated by distributed network terminals are collected and processed by edge servers. The edge servers split out a large sum of redundant data and transmit extracted information to the center cloud for further analysis. However, for consideration of limited edge computation ability, part of the raw data in huge data sources may be directly transmitted to the cloud. To manage limited resources online, we propose an algorithm based on Lyapunov optimization to jointly optimize the policy of edge processor frequency, transmission power and bandwidth allocation. The algorithm aims at stabilizing data processing delay and saving energy without knowing probability distributions of data sources. The proposed network management algorithm may contribute to big data processing in future IoT.

\end{abstract}


%
\IEEEpeerreviewmaketitle

\section{Introduction}
%
%
%
%
\IEEEPARstart{T}{he} internet of things (IoT) has emerged as a huge network, which extends connected agents beyond standard devices to any range of traditionally non-internet-enabled devices. In IoT, a large range of everyday objects such as vehicles, home appliances and street lamps may all enter the network and exchange data. The extension will result in an extraordinary increase of data amount and network cover range, which is far beyond the capability of the existing network. Recently, Mobile-edge computing (MEC) has emerged as a promising technique in IoT. By deploying cloud-like infrastructure in the vicinity of edge devices, a large proportion of computing load can be distributed to the edge \cite{7883946}.

In the literature, the problem of computation offloading, network resource allocation and related network structure designs in MEC have been broadly studied in various models \cite{he2018integrated,bi2018computation,rimal2017cloudlet,mao2017stochastic,zheng2019fog,park2016joint}. In \cite{he2018integrated}, the authors employed deep reinforcement learning to allocate caching, computing and communication resources for MEC system in vehicle networks. In \cite{bi2018computation}, the authors optimized the offload decision and resource allocation to obtain a maximum computation rate for a wireless powered MEC system. Considering the combination of MEC and existing communication service, a novel two-layer TDMA-based unified resource management scheme was proposed to handle both conventional communication service and MEC data traffic at the same time \cite{rimal2017cloudlet}. In \cite{mao2017stochastic}, the authors jointly optimized the radio and computational resource for Multi-user MEC computing system. In \cite{zheng2019fog}, notions of energy harvesting were further considered. In addition to the edge, the cloud was also taken into consideration in \cite{park2016joint}.

The MEC system design considering computation task offloading has been sufficiently investigated in previous works. However, for IoT big data processing, MEC server may also serves to process local data at the edge \cite{shi2016edge,shi2016promise,8336572}. In \cite{shi2016edge}, the authors discussed the application of MEC in data processing. In \cite{shi2016promise}, the authors indicated that edge servers can process part of the data rather than completely deliver them to the cloud. Then in \cite{8336572}, the authors proposed a scheme for this system. In the field of edge computing, the algorithm design for distributed data processing is still an open problem.

\begin{figure}
  \centering
  \includegraphics[width=0.8\columnwidth]{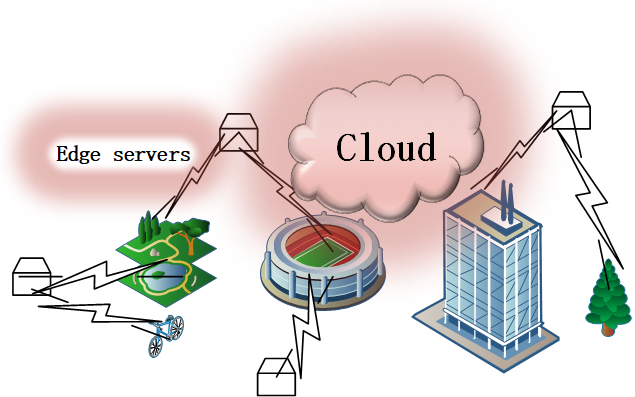}
  \caption{System structure of distributed edge data processing.}\label{structure}
\end{figure}

In this paper, we consider an MEC-based distributed data processing system as shown in Fig .\ref{structure}. In this system, servers at the network edge collect data from around data sources and conduct initial steps of data processing. Consider the common redundancy in raw data \cite{wan2019minor}, the edge processing will wipe out a large amount of redundant data and transmit extracted information to the cloud. It is assumed that the extracted information takes only a little bandwidth to transmit. As the edge processing speed is limited, part of the raw data will be transmitted to the cloud in cases of high data rate. As the communication resources are also limited, the rest data will be temporarily stored, which results in waiting delay.

Based on Lyapunov optimization, we proposed an algorithm to derive an online policy of network management. Without knowing probability distributions of arriving data, it can smartly manage network sources to stabilize delay while saving energy.
When data rate reduces, edge servers can lower down their processor frequency to save energy. In cases of high data rate, data offloading assists to raise edge processing speed. Furthermore, the allocation of bandwidth for data offloading can also adjust the edge processing capability based on their buffer lengths. In condition of high data rate, the smart design of bandwidth allocation can further stabilize edge processing delay. In order to figure out the policy design, we propose a network management algorithm based on Lyapunov optimization.

\section{System model}

We consider an IoT network for online data collection and analysis. The data sources are distributed in a wide range. The data are supposed to be generated randomly and transmitted to IoT edge servers. The distributed edge processing results are sent to center cloud $C$ for further analysis. The IoT network management policy is determined per time slot. The edge servers are represented by $E=\{ e_{k} \}$, where index $k$ belongs to set $\mathbb{K}=\{ 0, 1, 2, ......, K \}$ and the discrete time slot set is denoted as $\mathbb{T}=\{ 0, 1, 2, ...... \}$. In this section, we will introduce the model of data collection and processing.

\subsection{Data collection}
The widely distributed network devices generate data indicating local information. Edge servers collect data from their around devices. It is supposed that edge server $e_{k}$ collects $A_{k}(t)$ bits data during time slot $t$, where $k \in \mathbb{K}$ and $t \in \mathbb{T}$. The collected data will be temporarily stored in a data buffer for processing. Suppose the edge server $e_{k}$ is able to deal with $D_{k}(t)$ bits data in time slot $t$. Its data buffer length $Q_{k}(t)$ is updated by
\begin{equation}\label{qup}
Q_{k}(t+1)={\rm max}\{ Q_{k}(t)+A_{k}(t)-D_{k}(t), 0 \}
\end{equation}
It is assumed that $A_{k}(t)$ are independent among different devices and different time slots. $A_{k}(t)$ is supposed to satisfy poisson distribution with $E[A_{k}(t)] = \lambda_{k}$. Besides,for consideration of rate limitation in practical network, it is supposed that $A_{k}(t)$ is bounded by $[0, A_{max}]$. That is, any $A_{k}(t)$ larger than $A_{max}$ will be cut as $A_{max}$.

\subsection{Edge computation model}
It is assumed edge server $e_{k}$ has the capability to deal with $D_{k}(t)$ bits data in time slot $t$. Among the $D_{k}(t)$ bits data, $D_{l,k}(t)$ bits data are processed locally by edge server and $D_{tx,k}(t)$ bits data are transmitted to center. It is assumed that the edge servers will split out a large sum of redundant data and the extracted results take only a small proportion of bandwidth for transmission. Furthermore, the limited edge processing speed may not catch up with upcoming data rate. Then a large proportion of bandwidth can be allocated to $e_{k}$ for offloading data.

\subsubsection{Edge data processing}
It is assumed that the edge server $e_{k}$ needs $L_{k}$ CPU cycles to precess one bit data, which depends on the applied algorithm \cite{mao2017stochastic}. The processor cycle frequency of $e_{k}$ at time $t$ is denoted as $f_{k}(t)$ with $f_{k}(t) \in [0,f_{max}]$. Then $D_{l,k}(t)$ is
\begin{equation}
D_{l,k}(t)=\frac{\tau f_{k}(t)}{L_{k}}
\end{equation}
where $\tau$ is the time slot length. The power consumption rate of edge data processing \cite{burd1996processor} by $e_{k}$ is
\begin{equation}
P_{l,k}(t)=\kappa _{k}f_{k}^{3}(t)
\end{equation}
where $\kappa _{k}$ is the effective switched capacitance \cite{burd1996processor} of $e_{k}$, which is determined by chip structure.

\subsubsection{Data transmission model}
The edge data processing is limited by edge processor and energy resources. To lower down the delay, the network communication bandwidth is allocated to edge servers for transmission of collected data. It is assumed that the wireless channels between edge servers and center cloud are i.i.d. frequency-flat block fading \cite{lu2018beyond}. Thus the channel power gain between edge server $e_{k}$ and center cloud can be denoted by $\Gamma_{k}(t)=\gamma_{k}(t)g_{0}(\frac{d_{0}}{d_{k}})^{\theta}$, where $\gamma_{k}(t)$ is the small-scale fading channel power gain, $g_{0}$ is the pass loss constant, $\theta$ is the pass loss exponent, $d_{0}$ is reference distance and $d_{k}$ is the distance between $e_{k}$ and center cloud. Under the application of FDMA, by Shannon formula \cite{cover2006elements}, the data transmission capacity between $e_{k}$ and center cloud in time slot $t$ is
\begin{equation}
D_{tx,k}(t)=a_{k}(t)W\tau{\rm log}_{2}(1+\frac{\Gamma_{k}(t)p_{tx,k}(t)}{a_{k}(t)N_{0}W})
\end{equation}
where $a_{k}(t)$ is the proportion of the bandwidth allocated to $e_{k}$, $p_{tx,k}(t)$ is the transmission power with $p_{tx,k}(t)\in [0,p_{tx,max}]$, $W$ is the entire bandwidth for data transmission and $N_{0}$ is the noise power spectral density. $a_{t}=\{ a_{1}(t), a_{2}(t), ......, a_{K}(t) \}$ is the bandwidth allocation vector at time $t$ with $\sum_{i=1}^{K}a_{k}(t)=1$ and $a_{k}(t) \geq0$.

\section{Problem formulation}
The data offloading policy focus on the power consumption with respect to edge data processing and data transmission. In time slot $t$, the power consumption of edge processing of $e_{k}$ is denoted as $p_{l,k}(t)$. The data transmission power of $e_{k}$ in time slot $t$ is $p_{tx,k}(t)$. Then the power consumption of $e_{k}$ in time slot $t$ is
\begin{equation}
P_{k}(t)=p_{l,k}(t)+p_{tx,k}(t)
\end{equation}
Then the average weighted sum power consumption is
\begin{equation}
\overline{P}=\underset{T\rightarrow \infty }{lim}\frac{1}{T}\sum_{t=1}^{T}E\left [ \sum_{k=1}^{K}w_{k}P_{k}(t) \right ]
\end{equation}
where $w_{k}$ is a positive parameter with regard to edge server $e_{k}$, which can be adjusted to balance power management of all edge nodes. As the system performance metrics, $\overline{P}$ is the long-term edge power consumption. The data offloading policy with respect to $\overline{P}$ can be derived by statistical optimization.

The data collected by edge servers will be temporarily stored in a data buffer. In this case, the data queuing delay is the metrics of edge system service quality. By Little's Law \cite{little1961proof}, the average queuing delay of a queuing agent is proportional to the average queuing length. Therefore, the average data amount in edge data memory is viewed as the system service quality metrics. The long-term queuing length for edge server $e_{k}$ is defined as
\begin{equation}
\overline{Q}_{k}=\underset{T\rightarrow \infty }{lim}\frac{1}{T}\sum_{t=1}^{T}E[Q_{k}(t)]
\end{equation}

The network management policy in time slot $t$ is denoted as $\mathbf{\Phi }(t)=[\mathbf{f}(t),\mathbf{p}_{tx}(t),\mathbf{a}(t)]$. The operation set $\mathbf{f}(t)$ is the processor frequency of edge servers. The operation set $\mathbf{p}_{tx}(t)$ is the transmission power of data offloading. The parameter $\mathbf{a}(t)$ is the set of bandwidth allocation policy. Therefore, the optimal policy design is formulated as follows.

\begin{flalign}\label{equ:data1}
\mathcal{P}_{1}:\,\,\min_{\mathbf{\Phi }(t)} \,\,\, & \overline{P}\\
\textrm{s.t.}\,\,\,& \sum_{k=1}^{K}a_{k}(t)\leq 1 ,\, \, \, a_{k}(t)\geq \epsilon\,,k \in \mathbb{K},t \in \mathbb{T}.\tag{\theequation a}\label{equ:data1_a}\\
&0\leq f_{k}(t)\leq f_{max}, 0\leq p_{tx,k}(t)\leq p_{tx,max},\nonumber\\&k \in \mathbb{K},t \in \mathbb{T}.\tag{\theequation b}\label{equ:data1_b}\\
&\underset{T\rightarrow \infty }{lim}\frac{{\rm E}[\left | Q_{k}(t) \right |]}{T}=0\,,k \in \mathbb{K}.\tag{\theequation c}\label{equ:data1_c}
\end{flalign}
Eq .(\ref{equ:data1_a}) is the bandwidth allocation constraint, where $\epsilon$ is a system constant. Constraints (\ref{equ:data1_b}) indicates the bound of edge processor frequency and data transmission power. Index $k$ belongs to set $\mathbb{K}$ and time slot $t$ belongs to set $\mathbb{T}$. For delay consideration, constraint (\ref{equ:data1_c}) forces the edge data buffers to be stable, which guarantees the collected data can be processed in a finite delay.

The proposed $\mathcal{P}_{1}$ is obviously a statistical optimization problem with randomly arriving data. Therefore, the policy $\mathbf{\Phi }(t)$ has to be determined dynamically in each time slot. Furthermore, the spatial coupling of bandwidth allocation among edge servers induces great challenge to the problem solution. Instead of solving $\mathcal{P}_{1}$ directly, we propose an online jointly resource management algorithm based on Lyapunov optimization.

\section{Online network management}
\subsection{Lyapunov optimization framework}
The proposed $\mathcal{P}_{1}$ is a challenging statistical optimization problem. By Lyapunov optimization \cite{7274642}, $\mathcal{P}_{1}$ is formulated as a deterministic problem for each time slot, which can be solved with low complexity. The online algorithm can cope with the dynamical random environment while deriving an overall optimal outcome. Based on Lyapunov optimization framework ,the algorithm aims at saving energy while stabilizing the edge data buffers.

For online resource management, the Lyapunov function for each time slot is defined as
\begin{equation}
L(t)=\frac{1}{2}\sum_{k=1}^{K}Q_{k}^{2}(t)
\end{equation}
This quadratic function is a scalar measure of data accumulation in queue. Then the Lyapunov drift is defined as follows.
\begin{equation}
\Delta L(t)={\rm E}[L(t+1)-L(t)]
\end{equation}
To stabilize the network queuing buffer while minimizing the average energy penalty, the data processing policy is determined by minimizing a bound on the following drift-plus-penalty function for each time slot $t$.
\begin{equation}\label{dpp}
\Delta _{V}(t)=\Delta L(t)+V\sum_{k=1}^{K}w_{k}P_{k}(t)
\end{equation}
where $V$ is a positive system parameter which represents the tradeoff between Lyapunov drift and energy cost. $\Delta L(t)$ is the expectation of a random process whose probability distribution is supposed to be unknown. Therefore, an upper bound of $\Delta L(t)$ is estimated so that we can minimize $\Delta_{V}(t)$ without the specific probability distribution. According to the following Lemma \ref{upbound}, we derive a deterministic upper bound of $\Delta L(t)$ for each time slot.
\begin{lemma}\label{upbound}
For an arbitrary policy $\Phi (t)$ constrained by (\ref{equ:data1_a}), (\ref{equ:data1_b}) and (\ref{equ:data1_c}), the Lyapunov drift function is upper bounded by
\begin{equation}\label{upb}
\Delta L(t) \leq -\sum_{k=1}^{N_{U}}Q_{k}(t)(D_{l,k}(t)+D_{tx,k}(t))+C_{lp}
\end{equation}
where $C_{lp}$ is a known constant independent with the system policy and $Q_{k}(t)$ is the current data buffer length. $D_{l,k}(t)$ is the edge processing data bits amount while $D_{tx,k}(t)$ is the offloaded data amount. They are all for time slot $t$.
\end{lemma}
\begin{proof}
From equation (\ref{qup}), we have

\begin{align}
Q_{k}^{2}(t+1)& \leq (Q_{k}(t)+A_{k}(t)-(D_{l,k}(t)+D_{tx,k}(t)))^{2} \nonumber \\
&=Q_{k}^{2}(t)-2Q_{k}(t)(D_{l,k}(t)+D_{tx,k}(t)-A_{k}(t))+ \nonumber \\
&(D_{l,k}(t)+D_{tx,k}(t)-A_{k}(t))^{2} \label{drift1}
\end{align}
By (\ref{drift1}), we can subtract $Q_{k}^{2}(t)$ on both side and sum up the inequalities for $k=1, 2, ......, K$, which leads to follows.
\begin{align}
&\frac{1}{2}\sum_{k=1}^{K}\left [ Q_{k}^{2}(t+1)-Q_{k}^{2}(t) \right ] \nonumber \\
&\leq -\sum_{k=1}^{K}Q_{k}(t)(D_{l,k}(t)+D_{tx,k}(t))+\nonumber\\
&\sum_{k=1}^{K}\frac{(D_{l,k}(t)+D_{tx,k}(t)-A_{k}(t))^{2}}{2}+\sum_{k=1}^{K}Q_{k}(t)A_{k}(t)
\end{align}
It has been stated that $A_{k}(t)$ is bounded by $[0,A_{max}]$. Note that the computation and communication resources are limited, then $D_{l,k}(t)$ and $D_{tx,k}(t)$ are also bounded by their corresponding maximum rate. As the maximum processor frequency is $f_{max}$, we have $0 \leq D_{l,k}(t) \leq \frac{\tau f_{max}}{L_{k}}$. Since ${\rm log}_{2}(1+x) \leq \frac{x}{{\rm ln}2}$ and $p_{tx,k}(t) \in [0,p_{tx,max}]$, we have $0 \leq D_{tx,k}(t) \leq \frac{\tau}{N_{0}}p_{tx,max}\gamma_{k}g_{0}(\frac{d_{0}}{d_{k}})^{\theta}$. For simplicity, we separately denote $\frac{\tau f_{max}}{L_{k}}$ and $\frac{\tau}{N_{0}}p_{tx,max}\gamma_{k}g_{0}(\frac{d_{0}}{d_{k}})^{\theta}$ as $D_{l,k,max}$ and $D_{tx,k,max}$. Then the term $(D_{l,k}(t)+D_{tx,k}(t)-A_{k}(t))^{2}$ should be bounded by ${\rm max}\{ A_{max}^{2}, (D_{l,k,max}+D_{tx,k,max})^{2} \}$
Therefore, we have
\begin{align}
&\frac{1}{2}\sum_{k=1}^{K}\left [ Q_{k}^{2}(t+1)-Q_{k}^{2}(t) \right ] \nonumber \\
&\leq -\sum_{k=1}^{K}Q_{k}(t)(D_{l,k}(t)+D_{tx,k}(t))+\nonumber\\
&\sum_{k=1}^{K}\frac{{\rm max}\{ A_{max}^{2}, (D_{l,k,max}+D_{tx,k,max})^{2} \}}{2}+\sum_{k=1}^{K}Q_{k}(t)A_{k}(t)\nonumber\\
&= -\sum_{k=1}^{K}Q_{k}(t)(D_{l,k}(t)+D_{tx,k}(t))+C_{lp}
\end{align}
where $C_{lp}=\sum_{k=1}^{K}\frac{{\rm max}\{ A_{max}^{2}, (D_{l,k,max}+D_{tx,k,max})^{2} \}}{2}+\sum_{k=1}^{K}Q_{k}(t)A_{k}(t)$. When considering a specific time slot $t$, it is straightforward to see that $C_{lp}$ is a deterministic constant. This completes the proof.
\end{proof}
Together with (\ref{dpp}) and (\ref{upb}), the drift-plus penalty function is upper-bounded by
\begin{equation}
\Delta _{V}(t)\leq -\sum_{k=1}^{K}Q_{k}(t)(D_{l,k}(t)+D_{tx,k}(t))+V\sum_{k=1}^{K}w_{k}P_{k}(t)+C_{lp}
\end{equation}
By optimizing the above upper bound of $\Delta_{V}(t)$ in each time slot $t$, the data queuing length can be stabilized on a low level while the power consumption penalty is also minimized.
In Lemma \ref{upbound}, parameter $C_{lp}$ is not affected by system policy. Therefore, it is reasonable to omit $C_{lp}$ in the policy determination problem.

Then the modified problem is defined as follows.
\begin{flalign}\label{equ:data2}
\mathcal{P}_{2}:\,\,\min_{\mathbf{\Phi }(t)} \,\,\, & -\sum_{k=1}^{K}Q_{k}(t)(D_{l,k}(t)+D_{tx,k}(t))\nonumber\\&+V\sum_{k=1}^{K}w_{k}P_{k}(t)\\
\textrm{s.t.}\,\,\,& \sum_{k=1}^{K}a_{k}(t)\leq 1 ,\, \, \, a_{k}(t)\geq \epsilon\,,k \in \mathbb{K}\,\,, t \in \mathbb{T}.\tag{\theequation a}\label{equ:data2_a}\\
&0\leq f_{k}(t)\leq f_{max}, 0\leq p_{tx,k}(t)\leq p_{tx,max},\nonumber\\&k \in \mathbb{K}\,\,, t \in \mathbb{T}.\tag{\theequation b}\label{equ:data2_b}
\end{flalign}

%

\subsection{Solution for $\mathcal{P}_{2}$}
In last subsection, we formulated $\mathcal{P}_{2}$ for deriving optimal policy in each time slot. The optimization objectives include edge computation processor frequency $\mathbf{f}(t)$, data transmission power $\mathbf{p}_{tx}(t)$ and bandwidth allocation $\mathbf{a}(t)$. In this section, we will divide $\mathcal{P}_{2}$ into two subproblems and derive a solution for optimal policy.
\subsubsection{Optimal frequency for edge processor}
We first delete part of the objective function which is not related with $\mathbf{f}(t)$. Then it is straightforward to see that the subproblem with respect to $\mathbf{f}(t)$ is defined as follows.
\begin{flalign}\label{subp1}
\mathcal{P}_{3\text{-}\rm{A}}:\,\,\min_{\mathbf{f }(t)} \,\,\, & -\sum_{k=1}^{K}\frac{\tau Q_{k}(t)}{L_{k}}f_{k}(t)+V\sum_{k=1}^{K}w_{k}\kappa _{k}f_{k}^{3}(t)\\
\textrm{s.t.}\,\,\,&0\leq f_{k}(t)\leq f_{max}\,,k \in \mathbb{K}\,\,, t \in \mathbb{T}.\tag{\theequation a}\label{equ:subp1_a}
\end{flalign}

It is obvious to confirm that $\mathcal{P}_{3\text{-}\rm{A}}$ is a convex optimization problem. Furthermore, there is no coupling among elements in $\mathbf{f }(t)$. Therefore, the optimal processor frequency can be derived separately for each edge server. The stationary point of $\frac{\tau Q_{k}(t)}{L_{k}}f_{k}(t)+Vw_{k}\kappa _{k}f_{k}^{3}(t)$ is $\sqrt{\frac{\tau Q_{k}(t)}{3L_{k}w_{k}\kappa_{k} V}}$. In addition, the optimal processor frequency may also be the boundary $f_{max}$. Then the final solution is given by
\begin{equation}\label{fopt}
f_{k}^{*}(t)={\rm min} \{ f_{max}, \sqrt{\frac{\tau Q_{k}(t)}{3L_{k}w_{k}\kappa_{k} V}} \}\,\,\,(w_{k}>0,V>0)
\end{equation}


\subsubsection{Bandwidth allocation and data transmission power}
We reserve the elements with respect to $\mathbf{p}_{tx}(t)$ and $\mathbf{a}(t)$ and derive the following subproblem.
\begin{flalign}\label{subp2}
\mathcal{P}_{3\text{-}\rm{B}}:\,\,\min_{\mathbf{p_{tx} }(t),\mathbf{a }(t)} \,\,\, & -\sum_{k=1}^{K}Q_{k}(t)D_{tx,k}(t)+V\sum_{k=1}^{K}w_{k}p_{tx,k}(t)\\
\textrm{s.t.}\,\,\,& \sum_{k=1}^{K}a_{k}(t)\leq 1 ,\, \, \, a_{k}(t)\geq \epsilon\,,k \in \mathbb{K}\,\,, t \in \mathbb{T}.\tag{\theequation a}\label{equ:data2_a}\\
&0\leq p_{tx,k}(t)\leq p_{tx,max},k \in \mathbb{K}\,\,, t \in \mathbb{T}.\tag{\theequation b}\label{equ:data2_b}
\end{flalign}

In (\ref{subp2}), we have
\begin{equation}
D_{tx,k}(t)=a_{k}(t)W\tau{\rm log}_{2}(1+\frac{\Gamma_{k}(t)p_{tx,k}(t)}{a_{k}(t)N_{0}W})
\end{equation}
Note that this is a perspective function of $\widetilde{D}(p_{tx}(t))=W\tau{\rm log}_{2}(1+\frac{\Gamma_{k}(t)p_{tx,k}(t)}{N_{0}W})$ with $D_{tx,k}(t)=a_{k}(t)\widetilde{D}(p_{tx}(t)/a_{k}(t))$. It is straightforward to see that $\widetilde{D}(p_{tx}(t))$ is a concave function with respect to $p_{tx}(t)$. Then $D_{tx,k}(t)$ is a jointly concave function with respect to $a_{k}(t)$ and $p_{tx,k}(t)$. Therefore, $\mathcal{P}_{3\text{-}\rm{B}}$ is a convex optimization problem. Though it can be solved directly by conventional solvers, the dimensional curse may still be a large obstacle. In this paper, we employ an iterative algorithm to solve $\mathcal{P}_{3\text{-}\rm{B}}$ in a more efficient way.

Suppose the bandwidth allocation $\mathbf{a }(t)$ is fixed, a sub-problem can be derived as follows.
\begin{flalign}\label{subp3}
\mathcal{P}_{3\text{-}\rm{C}}:\,\,\min_{\mathbf{p_{tx} }(t)} \,\,\, & -\sum_{k=1}^{K}Q_{k}(t)a_{k}(t)W\tau {\rm log}_{2}(1+\frac{\Gamma_{k}(t)p_{tx,k}(t)}{a_{k}(t)N_{0}W})\nonumber\\&+V\sum_{k=1}^{K}w_{k}p_{tx,k}(t)\\
\textrm{s.t.}\,\,\,&0\leq p_{tx,k}(t)\leq p_{tx,max},k \in \mathbb{K}\,\,, t \in \mathbb{T}.\tag{\theequation a}\label{equ:subp3_a}
\end{flalign}
As $\mathbf{a }(t)$ is fixed, $p_{tx,k}(t)$ in $\mathbf{p_{tx} }(t)$ are independent with each other. Therefore, we can separately obtain $p_{tx,k}^{*}(t)$ with respect to each index $k$. The stationary point of system cost function is $a_{k}(t)W[\frac{Q_{k}(t)\tau}{Vw_{k}{\rm ln}2}-\frac{N_{0}}{\Gamma_{k}(t)}]$. Considering constraint (\ref{equ:subp3_a}), the optimal solution of $\mathcal{P}_{3\text{-}\rm{C}}$ is
\begin{equation}\label{ptxopt}
p_{tx,k}^{*}(t)={\rm min}\{ {\rm max}\{ a_{k}(t)W[\frac{Q_{k}(t)\tau}{Vw_{k}{\rm ln}2}-\frac{N_{0}}{\Gamma_{k}(t)}],0 \},p_{tx,max} \}
\end{equation}

Suppose $\mathbf{p_{tx} }(t)$ is figured out, a sub-problem to optimize $\mathbf{a}(t)$ is derived as follows.
\begin{flalign}\label{subp4}
\mathcal{P}_{3\text{-}\rm{D}}:\,\,\min_{\mathbf{a}(t)} \,\,\, & -\sum_{k=1}^{K}Q_{k}(t)a_{k}(t)W\tau {\rm log}_{2}(1+\frac{\Gamma_{k}(t)p_{tx,k}(t)}{N_{0}Wa_{k}(t)})\\
\textrm{s.t.}\,\,\,&\sum_{k=1}^{K}a_{k}(t)\leq 1 ,\, \, \, a_{k}(t)\geq \epsilon\,,k \in \mathbb{K}\,\,, t \in \mathbb{T}.\tag{\theequation a}\label{equ:subp4_a}
\end{flalign}
In this case, elements in $\mathbf{a}(t)$ are coupled with each other. Therefore, we employ dual decomposition \cite{xiao2004simultaneous}.
For $\mathcal{P}_{3\text{-}\rm{D}}$, the Lagrange function is
\begin{flalign}
L(\mathbf{a}(t),\lambda)&=-\sum_{k=1}^{K}Q_{k}(t)a_{k}(t)W\tau {\rm log}_{2}(1+\frac{\Gamma_{k}(t)p_{tx,k}(t)}{N_{0}Wa_{k}(t)})\nonumber\\&+\lambda(\sum_{k=1}^{K}a_{k}(t)-1)
\end{flalign}
To decouple $a_{k}(t)$, we set
\begin{flalign}\label{dual-d}
g_{k}(\lambda)=&\underset{a_{k}(t)\geq \epsilon }{inf}(-Q_{k}(t)a_{k}(t)W\tau {\rm log}_{2}(1+\frac{\Gamma_{k}(t)p_{tx,k}(t)}{N_{0}Wa_{k}(t)})\nonumber\\&+\lambda a_{k}(t))
\end{flalign}

Then we have $L(\lambda)=\sum_{k=1}^{K}g_{k}(\lambda)-\lambda$, the dual sub-problem of $\mathcal{P}_{3\text{-}\rm{D}}$ is

\begin{flalign}\label{subp5}
\mathcal{P}_{3\text{-}\rm{E}}:\,\,\max_{\lambda}\,\,\, & \sum_{k=1}^{K}g_{k}(\lambda)-\lambda\,\,\, & \\
\textrm{s.t.}\,\,\,&\lambda \geq 0.\tag{\theequation a}\label{equ:subp5_a}
\end{flalign}
Dual sub-problem $\mathcal{P}_{3\text{-}\rm{E}}$ can be solved by gradient decent method. The corresponding gradient is $\sum_{k=1}^{K}a_{k}^{*}(t)-1$, where $a_{k}^{*}(t)$ achieves the lower bound in (\ref{dual-d}). As stated before, this is a convex optimization problem. Therefore, either the stationary point or $\epsilon$ achieves the lower bound. The value of stationary point is the zero point of the derivative function, which can be derived by bisection method. By iteratively updating $\lambda$ and corresponding $a_{k}^{*}(t)$, we can finally derive the optimal bandwidth allocation and transmission power.

In summary, the optimal policy for online edge data processing can be illustrated by chart in Fig .\ref{online_alg}.
\begin{figure}
  \centering
  \includegraphics[width=0.94\columnwidth]{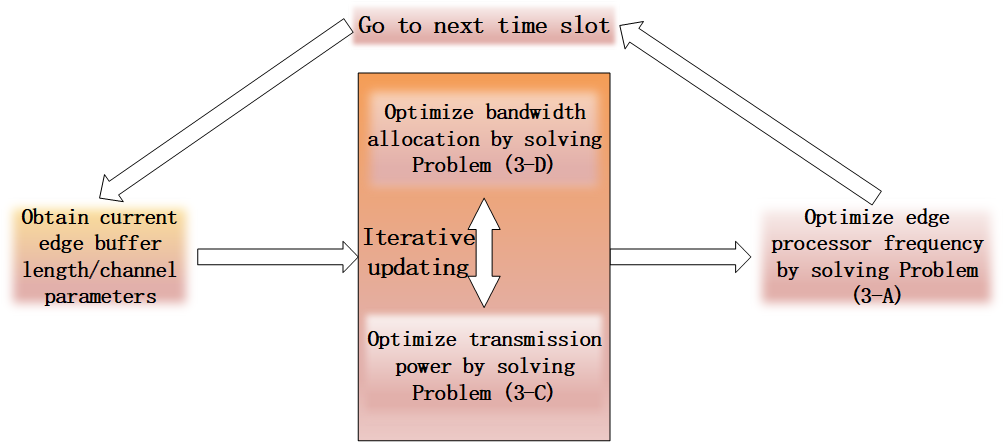}\\
  \caption{Process of network management for online edge data processing}\label{online_alg}
\end{figure}

\section{Simulation results}
We consider simulations of a network composed of $K$ edge servers and a center cloud. It is assumed that the center cloud has an equal distance with $K$ edge servers, which is set as 200. $A_{k}(t)$ satisfies poisson distribution with rate $\lambda_{k}$. Besides, for consideration of maximum data collection speed in real system, $A_{k}(t)$ is constrained in $[0,A_{max}]$. In simulations, if the randomly generated $A_{k}(t)$ is larger than $A_{max}$, it will be set as $A_{max}$. The small scale fading channel power gain $\gamma_{k}(t)$ is generated by exponential distribution $Exp(1)$. Besides, other parameter sets include $\tau=0.5s$, $W=2 MHz$, $N_{0}=-167 dBm/Hz$, $g_{0}=-40 dB$, $\theta=4$, $w_{k}=1/K$, $d_{0}=1m$, $f_{max}=2 GHz$, $\kappa_{k}=10^{-26}$, $p_{tx,max}=5 W$, $K=7$, $L_{k}=3000 cycles/bit$ and $\epsilon=10^{-3}$.

\begin{figure}[tbp]
\centering
\subfigure[] { \label{Fig:pv}
\includegraphics[width=0.45\columnwidth]{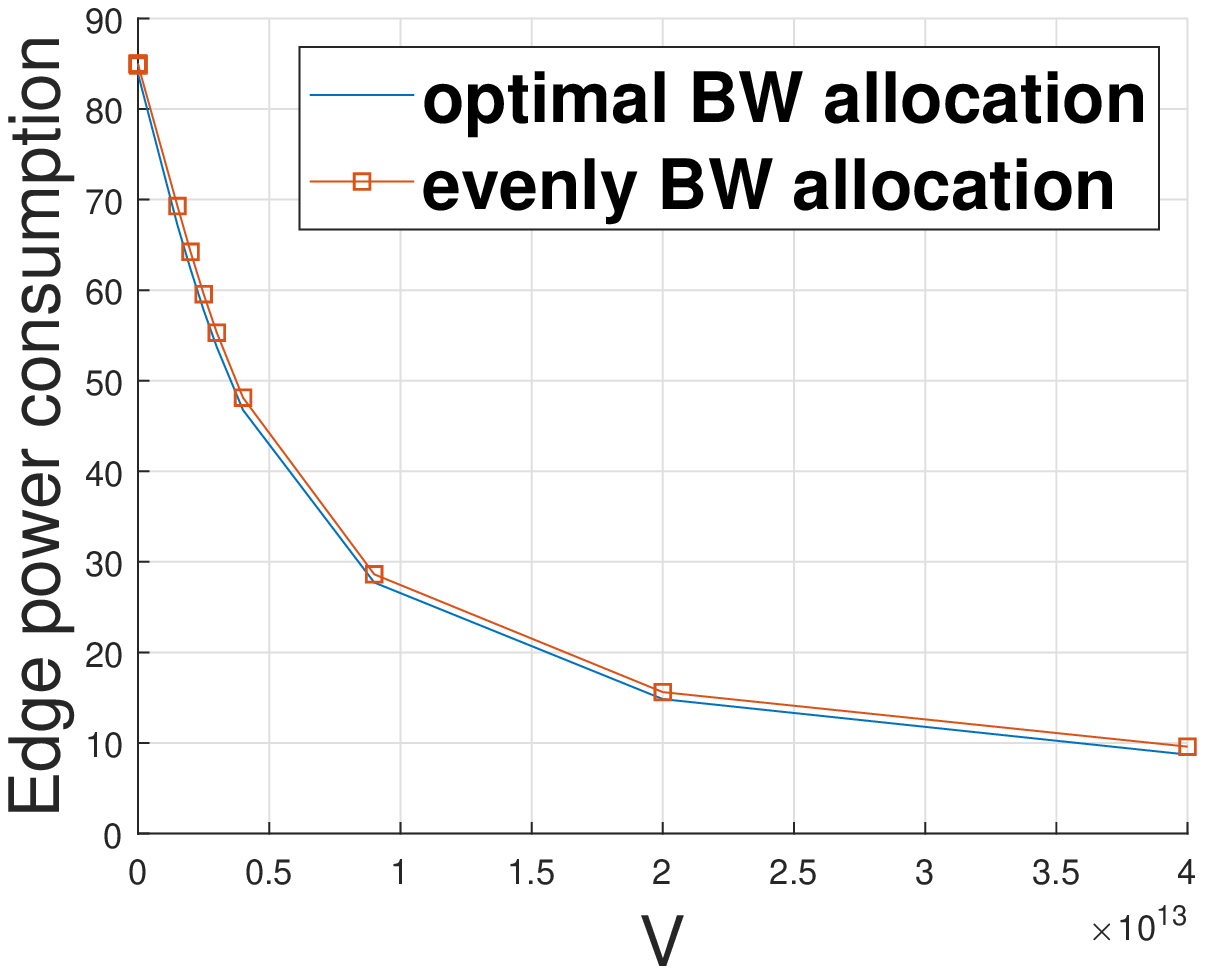}}
\subfigure[]{ \label{Fig:lv}
\includegraphics[width=0.45\columnwidth]{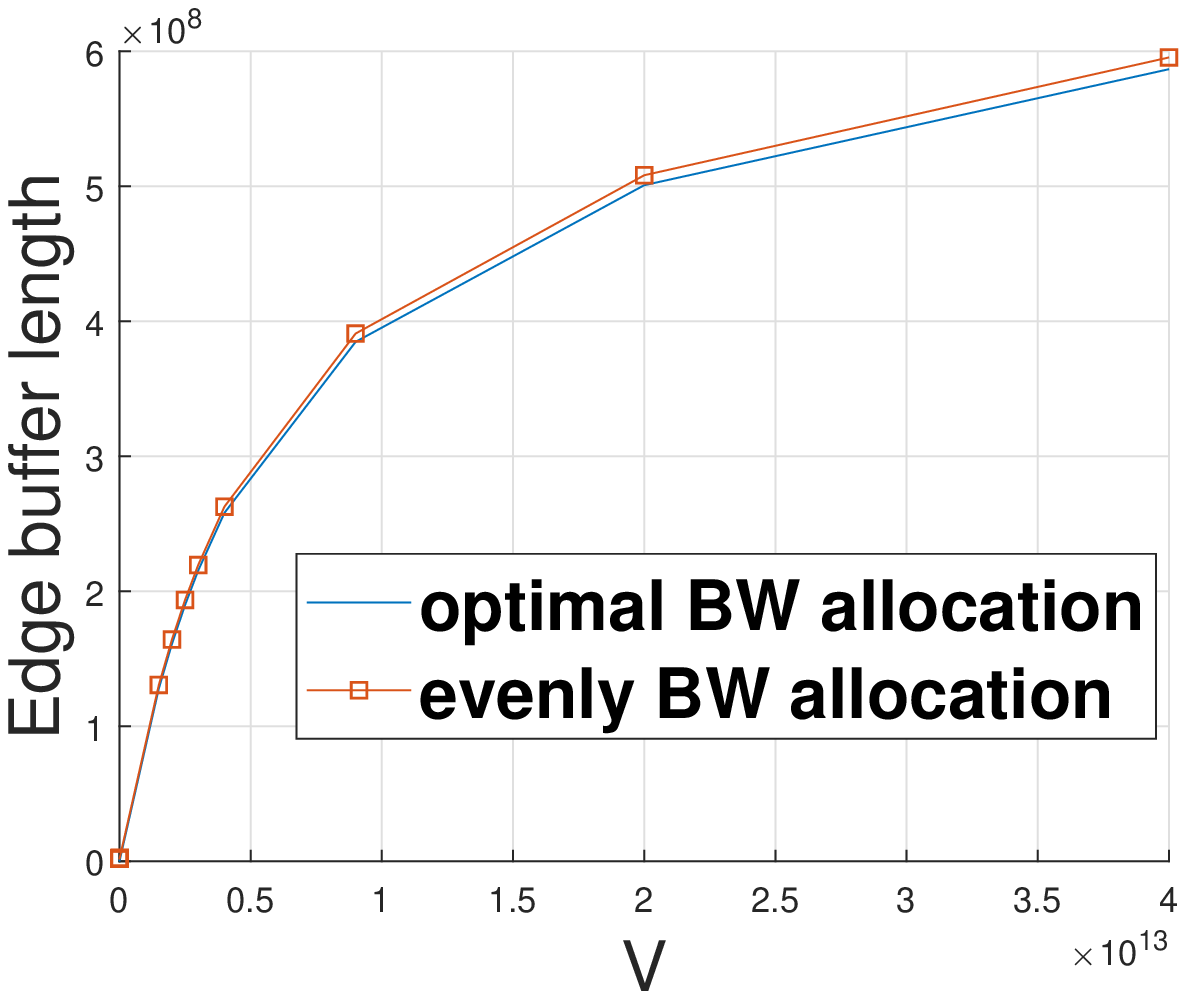}}
\subfigure[] { \label{Fig:pr}
\includegraphics[width=0.45\columnwidth]{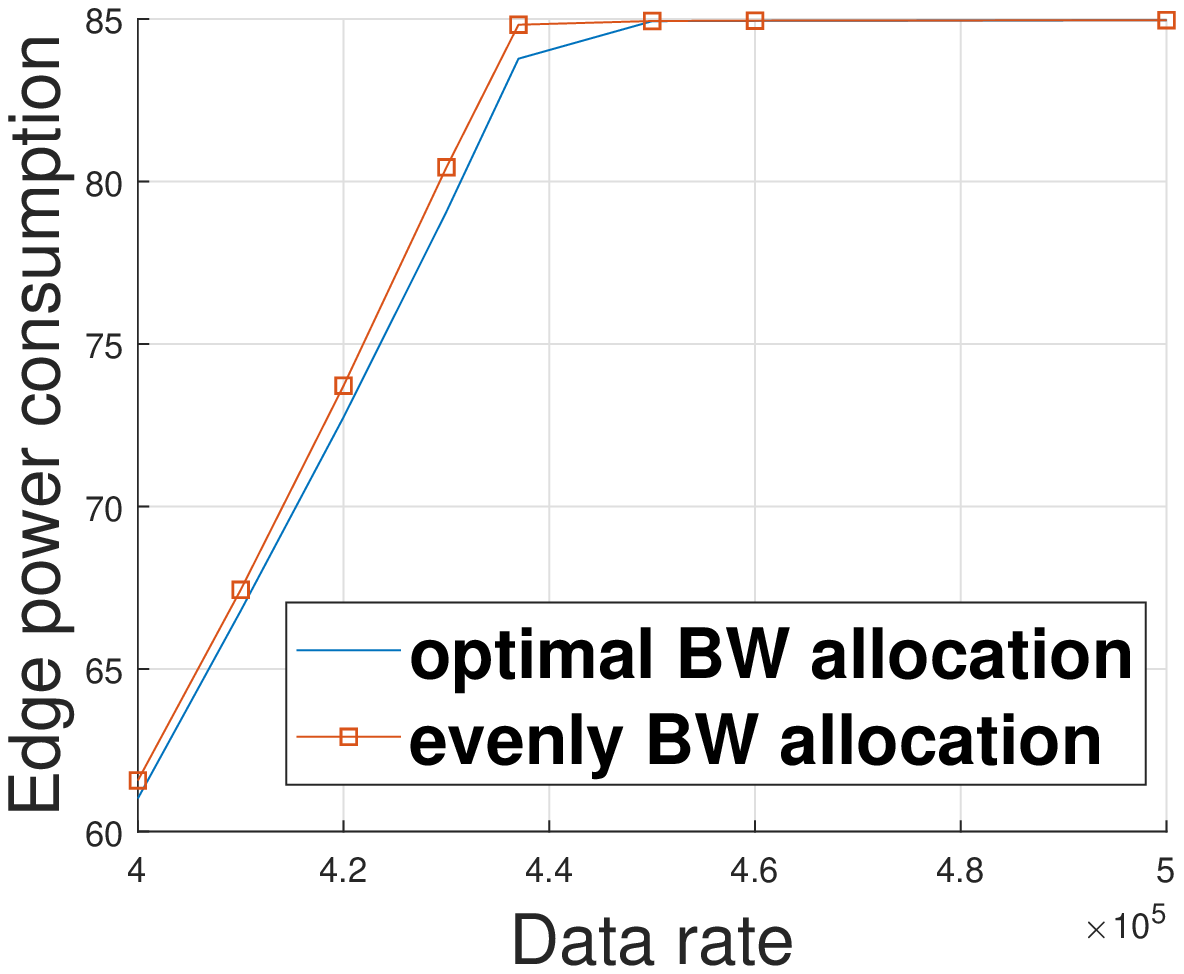}}
\subfigure[]{ \label{Fig:lr}
\includegraphics[width=0.45\columnwidth]{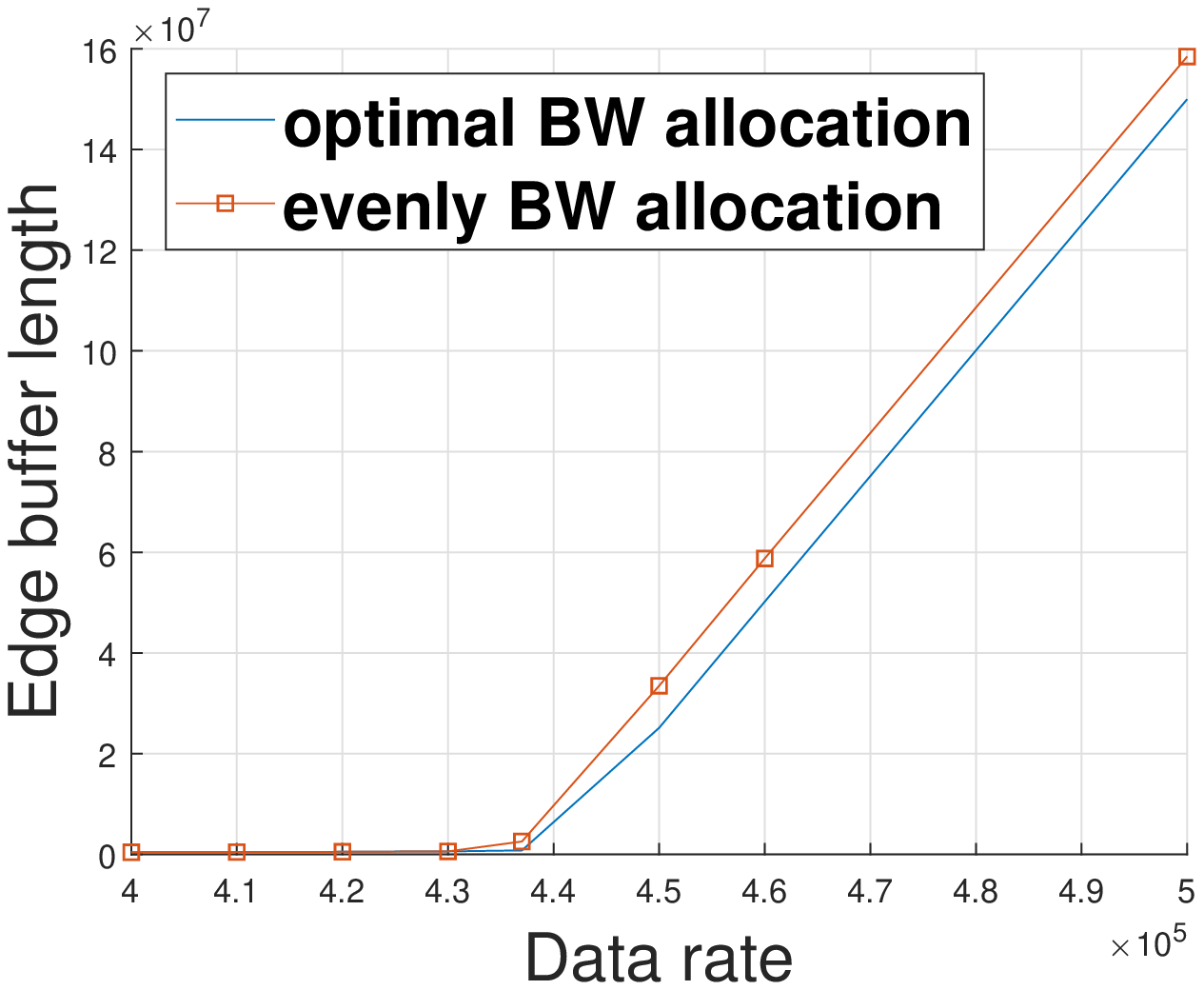}}
\caption{Edge power consumption and buffer length vs $V$ and data rate. Date rate of edges are equal with $\lambda_{k}=\lambda$. Sub-figures (a) and (b) show energy cost and buffer length vs control parameter $V$, where data rate $\lambda=4.37\times 10^{5}$. Sub-figures (c) and (d) show energy cost and buffer length vs data rate, where $V=10^{10}$. The results are derived by averaging records within 5000 time slots. } \label{cost}
\end{figure}

\begin{figure}
  \centering
  \includegraphics[width=0.9\columnwidth]{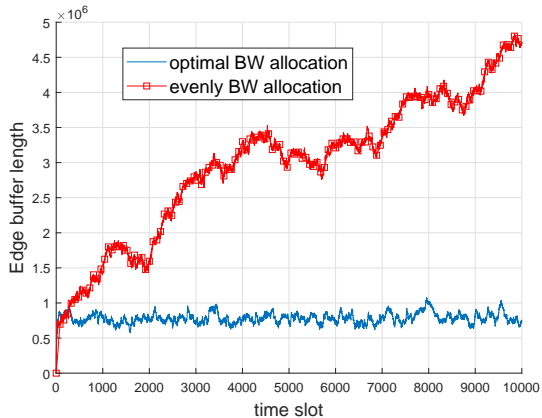}
  \caption{Edge data buffer length vs time slot for optimal bandwidth allocation and evenly bandwidth allocation ($\lambda_{k}=4.37\times 10^{5}, V=10^{10}$).}\label{bfl}
\end{figure}

The system performance in terms of power consumption and edge buffer length is first tested by two network management strategies. The results are shown in Fig .\ref{cost}. The curves marked by squares are obtained by evenly allocating bandwidth. Except for $a_{k}(t)=\frac{1}{K}$, its $f_{k}(t)$ and $p_{tx,k}(t)$ are both optimized. That is, the bandwidth allocation is separated from computation resource management. Curves without marks are obtained by optimizing bandwidth allocation. Fig .\ref{Fig:pv} shows that the average power consumption monotonically decreases with respect to control parameter $V$. Fig .\ref{Fig:lv} shows that the average edge buffer length monotonically increases with respect to $V$. By (\ref{fopt}) and (\ref{ptxopt}), the increase of $V$ will decrease $f_{k}(t)$ and $p_{tx,k}(t)$, which reduces energy consumption while lowering down the edge processing speed. Meanwhile, as shown in Fig .\ref{Fig:pr} and Fig .\ref{Fig:lr}, increasing data rate results in the rise of power consumption and edge buffer length. However, the performance deteriorates when taking evenly bandwidth allocation. This shows the importance of jointly optimizing bandwidth allocation and computation resources.

Fig .\ref{bfl} shows the average edge buffer length with respect to time. The curve marked by square is obtained by evenly bandwidth allocation. In cases of high data rate, the strategy with optimal bandwidth allocation achieves a stable edge buffer length. By Little's Law, this indicates a stable data processing delay, which is crucial for online data processing. However, the strategy with evenly bandwidth allocation shows an awful performance. In cases of high data rate with randomly arriving data amount, optimal bandwidth allocation tend to allocate more resources to edges with larger data amount. This explains its performance of stabilizing edge buffer length. Therefore, it is important to jointly consider bandwidth allocation in network management.

\section{Conclusion}
In this paper, we investigated network management strategies for online edge data processing in IoT. Focused on saving energy while stabilizing delay, an online MEC-based network management algorithm was proposed based on Lyapunov optimization framework. In cases of low data rate, edge processor frequency and transmission power are dynamically reduced for saving energy. In cases of high data rate, the bandwidth resources are optimally allocated for stabilizing data processing delay. By theoretical analysis and simulation tests, we validated the performance of the proposed dynamical network management algorithm with respect to system design parameters. The online policy is obtained by current data buffer length regardless of data source probability distributions.


%

%



\ifCLASSOPTIONcaptionsoff
  \newpage
\fi

\end{document}